\documentclass{tlp}

\newif\ifappendix
\appendixtrue
\ifdefined\appendixoff\appendixfalse\fi

\usepackage{amsmath}
\usepackage{amssymb}
\usepackage{amsfonts}
\usepackage{graphicx}
\usepackage{multirow}
\usepackage{url}
\usepackage{listings}
\usepackage{xcolor}
\usepackage{xspace}
\usepackage{mathtools}
\usepackage{hyperref}
\usepackage{subcaption}
\usepackage{tikz}
\usetikzlibrary{positioning,arrows.meta}
\usepackage{wrapfig}

\lstdefinestyle{logica}{
  basicstyle=\ttfamily\small,
  breaklines=true,
  keywords={Min=, +=, =},
  keywordstyle=\bf,
  comment=[l]{\#},
  commentstyle=\color{gray}\itshape,
}

\newtheorem{theorem}{Theorem}
\newtheorem{lemma}{Lemma}
\newtheorem{definition}{Definition}
\newtheorem{observation}{Observation}

\newtheorem{corollary}{Corollary}

\makeatletter
\def\@begintheorem#1#2{\intheoremtrue\normalfont\rmfamily\trivlist
  \pagebreak[3]\item[\hskip\labelsep{\normalfont\bfseries #1\ #2.}]}
\def\@opargbegintheorem#1#2#3{\intheoremtrue\normalfont\rmfamily\trivlist
  \pagebreak[3]\item[\hskip\labelsep{\normalfont\bfseries #1\ #2\ {\upshape(}#3\/{\upshape)}.}]}
\makeatother

\newcommand{\mypara}[1]{\vspace{3pt}\noindent\textbf{#1.}}

\newcommand{\mynote}[1]{} 

\newcommand{\madiamond}{\ensuremath{\mathord{\vcenter{\hbox{\scalebox{0.8}{\rotatebox[origin=c]{45}{$\Box$}}}}}\mkern1.33mu}}
\newcommand{\mabox}{\ensuremath{\mathord{\vcenter{\hbox{$\Box$}}}\mkern1.33mu}}
\newcommand{\bdb}{\ensuremath{\mabox\madiamond\mabox}}

\sbox{\proofsavebox}{\textcolor{gray!60}{\rule{6pt}{6pt}}}
\makeatletter
\renewenvironment{proof}
  {\@ifnextchar[{\@oprf}{\@nprf}}
  {\hfill\hbox{\proofbox}\endtrivlist}
\makeatother

\begin{document}

\lefttitle{Skvortsov et al.}

\jnlPage{1}{16}
\jnlDoiYr{2026}
\doival{10.1017/xxxxx}

\title[Diamonds Are Forever]{Diamonds Are Forever: Stabilization Semantics for Unrestricted Aggregation and Recursion in Logica}

\begin{authgrp}
\author{\gn{Evgeny} \sn{Skvortsov}}
\affiliation{Google LLC, USA}
\author{\gn{Yilin} \sn{Xia}}
\affiliation{University of Illinois, Urbana-Champaign, USA}
\author{\gn{Ojaswa} \sn{Garg}}
\affiliation{Google LLC, USA}
\author{\gn{Shawn} \sn{Bowers}}
\affiliation{Gonzaga University, USA}
\author{\gn{Bertram} \sn{Lud\"{a}scher}}
\affiliation{University of Illinois, Urbana-Champaign, USA}
\end{authgrp}

\history{\sub{xx xx 2026;} \rev{xx xx 2026;} \acc{xx xx 2026}}
\maketitle

\begin{abstract}
Logica is an open-source logic programming language that compiles to SQL and
runs on platforms including DuckDB, SQLite, PostgreSQL, and BigQuery. Unlike
classic Datalog, Logica permits the free combination of recursion and aggregation,
enabling concise formulations of algorithms from shortest paths to PageRank.
This expressiveness introduces fundamental semantic challenges: aggregate
predicates are updated by replacement rather than accumulation, evaluation is
sensitive to rule scheduling, and programs may converge to meaningful results
without ever reaching a fixpoint, placing them outside the scope of
traditional fixpoint semantics.

We address these challenges with \emph{Defendant-Opponent (DO) semantics}, a
stabilization-based framework for nonmonotonic logic programs. Program
evaluation is modeled as a rewrite system over derivation states. A ground atom
is true if, from every reachable state, there exists a continuation after which
the atom persists in all further derivations. This notion admits two equivalent
characterizations: (1)~game-theoretically, truth is what a Defendant can defend
against any Opponent in a three-turn derivation game; and (2)~modally, a
formula~$t$ is a theorem precisely when the condition $\bdb\,t$ holds in the
derivation graph viewed as a Kripke structure.

We show that DO semantics coincides with classic least fixpoint
semantics for positive Datalog programs and is compatible with both the
Well-Founded Semantics (the two never disagree on definite answers)
and the Stable Model Semantics (every stable model is a DO
interpretation). For programs that converge without reaching a
fixpoint, we introduce $\omega$-limit interpretations, giving rigorous
meaning to iterative computations such as PageRank. DO semantics thus
offers a coherent framework that complements existing logic
programming semantics while supporting recursive aggregation.
\end{abstract}

\noindent\textbf{Keywords:} Logica, Datalog, Nonmonotonic Reasoning, Aggregation, Modal Logic


\section{Introduction}

Logic programming has long served as a cornerstone of declarative computing. Datalog became particularly influential due to its elegant formal semantics grounded in monotonic least fixpoint computations~\citep{vanEmdenKowalski1976,ceri_logic_1990}, enabling concise recursive queries such as transitive closure.

Extending Datalog with negation led to fundamental semantic challenges. Well-Founded Semantics (WFS)~\citep{van1991well,vanGelder1993wfs} uses three-valued logic and an alternating fixpoint construction to handle nonmonotonic negation. Stable Model Semantics (SMS)~\citep{gel1988stable} takes an alternative approach through self-supporting fixpoints, enabling multiple stable interpretations and forming the basis of Answer Set Programming~\citep{Gelfond08,gebser_multi-shot_2019}.

Integrating aggregation into recursive logic programming brought further challenges~\citep{mumick1990aggregation,ross1997monotonic}. Work on stratified aggregates~\citep{kemp1991}, recursive aggregation~\citep{zaniolo_fixpoint_2017}, unified founded semantics~\citep{liu_recursive_2022}, and recursive aggregates in ASP~\citep{faber2004recursive} has made significant progress, but has primarily focused on programs that reach fixpoints. The distinctive challenge lies with programs that \emph{converge without reaching a fixpoint}.

PageRank~\citep{BRIN1998107} is a canonical example: iteration quickly approaches the eigenvector distribution, yet formally no fixpoint is attained. Such programs fall outside the scope of fixpoint-based semantics.

\mypara{Logica and the Semantic Gap}
Logica~\citep{Skvortsov_EDBT_2024,Skvortsov_Datalog2_2024}, an open-source logic programming language developed at Google, compiles to standard SQL and runs on platforms including DuckDB, SQLite, PostgreSQL, and BigQuery~\citep{skvortsov_full_2023}. Unlike classic Datalog, Logica allows unrestricted combinations of recursion and aggregation, addressing practical needs in modern data analytics. This exposes a semantic gap: programs like PageRank are expressible and useful, yet lack formal meaning under traditional fixpoint semantics.

\mypara{DO Semantics}
We introduce \emph{Defendant-Opponent (DO) semantics}, a framework that generalizes classical logic programming semantics to address these limitations. Inspired by Hilbert-style formal systems but extended for nonmonotonic programs, DO semantics defines truth through an adversarial interaction between a \emph{Defendant} (asserting statements) and an \emph{Opponent} (challenging assertions), giving meaning to programs with arbitrary recursion, aggregation, numerical computation, and non-stratified negation.

DO semantics also has a modal logic characterization. In monotonic systems, a formula $t$ is a theorem if it is derivable from axioms $A$, i.e., $A \Vdash \madiamond\, t$.
In nonmonotonic systems where derived facts can be retracted, mere derivability is insufficient. We show that DO-truth corresponds to the modal formula $A \Vdash \bdb\, t$: from any state, it is \emph{always possible} ($\mabox\madiamond$) to reach a state where $t$ holds, and once reached, $t$ \emph{necessarily persists} ($\mabox$) under further derivation. This places nonmonotonic reasoning within S4 modal logic, where the seven distinct modalities provide a complete taxonomy of formula stability.

DO semantics bears a family resemblance to dialogue games for intuitionistic and classical logic~\citep{lorenzen1961}, where truth emerges from a winning strategy in a two-player debate. However, the frameworks differ fundamentally. In Lorenzen games, players argue about the structure of a single formula through attack and defense moves on connectives. In DO semantics, the players jointly construct a \emph{derivation} as a trajectory through the space of derivation states. The Opponent challenges not internal formula structure, but the \emph{robustness} of derivability: can the Defendant reach a state containing the target atom that remains stable under any further moves? This shift from formula-structure to derivation-trajectory games is what allows DO semantics to handle nonmonotonic reasoning.

\mypara{Contributions}
The main contributions of this paper are:
\begin{enumerate}[leftmargin=*,topsep=2pt,itemsep=1pt]
\item A stabilization-based semantic framework (DO semantics) for nonmonotonic logic programs with unrestricted recursion and aggregation, with equivalent game-theoretic and modal ($\bdb$) characterizations.
\item Backward compatibility: DO semantics coincides with classical least fixpoint semantics for positive Datalog (Corollary~\ref{cor:datalog}).
\item Compatibility with WFS (Theorem~\ref{thm:wfs}) and SMS (Theorem~\ref{thm:stable}): the three semantics never disagree on definite answers, and every stable model is a DO interpretation.
\item $\omega$-limit interpretations that give rigorous meaning to convergent programs such as PageRank (Theorems~\ref{thm:pagerank} and~\ref{thm:sincos}).
\end{enumerate}

\mypara{Outline}
Section~\ref{section:logica} introduces Logica through examples. Section~\ref{sec:do_semantics} develops DO semantics, from monotonic formal systems to nonmonotonic rewriting rules, the Thesis Defense game, and modal logic. Section~\ref{sec:interpretations} introduces static, dramatic, and $\omega$-limit interpretations. Section~\ref{sec:connections} establishes compatibility with WFS and SMS.


\section{Logica}\label{section:logica}

This section introduces the fragment of Logica relevant to this paper. Many features (named arguments, meta-programming, database integration) are practical conveniences that do not require the semantic treatment we develop here. A complete formal syntax of this fragment, together with a glossary of the key terms used throughout, is provided in the online supplementary material.

In Logica, predicates begin with capital letters, variables begin with
lowercase letters, and rules end with semicolons.
\begin{lstlisting}[style=logica,basicstyle=\small\ttfamily,columns=fullflexible]
    Sibling(a,b) :- Parent(x,a), Parent(x,b), a!=b ;
    Ancestor(a,b) :- Parent(a,b);
    Ancestor(a,b) :- Ancestor(a,x), Ancestor(x,b);
\end{lstlisting}
Logica supports functional notation: $P(x) = y$ is syntactic sugar for the relation $P(x, y)$ where the last column is viewed as the output.
\begin{lstlisting}[style=logica,basicstyle=\small\ttfamily,columns=fullflexible]
    Square(x) = x * x;
    Hypotenuse(a,b) = Sqrt(Square(a) + Square(b));
\end{lstlisting}
The central extension relevant to this paper is aggregation, which comes in two forms. \emph{Aggregating expressions} appear within rule bodies (similar to ASP syntax \cite{Gelfond08,gebser_multi-shot_2019}):
\begin{lstlisting}[style=logica,basicstyle=\small\ttfamily,columns=fullflexible]
    HighDegree(x) :- E(x,y), Count{z :- E(x,z)} > 5;
\end{lstlisting}
\emph{Predicate-level aggregation} resembles SQL's \texttt{GROUP BY} and aggregates over all derivations of a predicate's value. In this paper, we focus on Logica's predicate-level function value aggregation:
\begin{definition}
A \emph{rule with predicate-level function value aggregation} has the form
$$P(t_1, \ldots, t_n) \ \odot\!= e \ \texttt{:-}\ \mathit{Body}(t_1, \ldots, t_n, t_{n+1}, \ldots, t_m, e)$$
where $\odot$ is an aggregation operator that maps finite multisets to a single value; for each tuple $(t_1, \ldots, t_n)$, $\odot$ aggregates the values of $e$ over all instantiations satisfying $\mathit{Body}$.
\end{definition}
In practice, efficient aggregation requires certain algebraic properties (namely, an Abelian semigroup operator), but this is irrelevant for the semantic treatment; we use \texttt{Sum}, \texttt{Min}, and \texttt{Max} as running examples.

As a simple example, the shortest distances between all pairs of vertices in a graph can be computed in Logica as follows:
\begin{lstlisting}[style=logica,basicstyle=\small\ttfamily,columns=fullflexible]
   D(a,b) Min= E(a,b); 
   D(a,b) Min= D(a,x) + D(x,b); 
\end{lstlisting} 
For each pair $(a, b)$, the two rules define a multiset of candidate distances: the direct edge weight $E(a, b)$ if it exists, together with sums $D(a, x) + D(x, b)$ over all intermediate vertices~$x$. The \texttt{Min=} operator returns the minimum of this multiset, which is assigned as the value of $D$. As we see in Section~\ref{sec:logica-is-nmfs}, the computation of the distance amounts to iterated application of these rules.

PageRank~\citep{BRIN1998107} provides a more challenging example:
\begin{lstlisting}[style=logica,basicstyle=\small\ttfamily,columns=fullflexible]
    PageRank(x) += ResetProb() / N() :- Page(x);
    PageRank(y) += (1.0 - ResetProb()) * PageRank(x) / Degree(x) :- Link(x,y);
\end{lstlisting}
Iteration converges to the eigenvector distribution, yet formally never reaches a fixpoint: each step refines the approximation without stabilizing. Traditional fixpoint semantics cannot assign meaning to this program, which is precisely the gap that DO semantics fills. Logica can also express numerical computations. The following approximates $\pi$ using the Leibniz formula:
\begin{lstlisting}[style=logica,basicstyle=\small\ttfamily,columns=fullflexible]
    L() = 1 :- L = nil;
    L() = L() + 1;
    Pi() = 4 * Sum{(-1) ^ k / (2 * k + 1) :- k in Range(L())};
\end{lstlisting}
The \texttt{in} operator iterates over a collection. The following computes sine and cosine via Euler's method for differential equations.
\begin{lstlisting}[style=logica,basicstyle=\small\ttfamily,columns=fullflexible]
    Infinitesimal() = 1 / L();
    C(0) = 1; 
    S(0) = 0;  
    S(t + dt) = S(t) + C(t) * dt :- dt = Infinitesimal();
    C(t + dt) = C(t) - S(t) * dt :- dt = Infinitesimal();
\end{lstlisting}
%
DO semantics makes this precise through $\omega$-limit interpretations (Section~\ref{sec:interpretations}).

Finally, negation in Logica is encoded via aggregation: \texttt{not\;$P(\vec{x})$} becomes \texttt{Count\{1\,:-\,$P(\vec{x})$\}\;=\;0}, which holds precisely when $P(\vec{x})$ has no derivation in the current state. This reduces negation to aggregation, requiring no additional semantic machinery.

\section{Defendant-Opponent (DO) Semantics}
\label{sec:do_semantics}

We define DO semantics by generalizing Hilbert-style formal systems to nonmonotonic logic. We begin by recalling the standard monotonic case.

\subsection{Monotonic Formal Systems}

\begin{definition}
\label{def:formal-system}
A \emph{formal system} is a triple $(H, A, R)$, where $H$ is a set of \emph{well-formed formulas}, $A \subseteq H$ is a set of axioms, and $R$ is a set of inference rules. Each rule $r \in R$ is applied to a tuple $(t_1, \dots, t_n)$ of well-formed formulas and returns a well-formed formula $r(t_1, \dots, t_n)$ that follows from the rule, i.e., $r: H^n \rightarrow H$.
\end{definition}
In logic programming, $H$ corresponds to the Herbrand base, i.e., the set of all potentially derivable ground atoms.

\begin{definition}
\label{def:derivation-mono}
For a formal system $(H, A, R)$, a sequence $T_0, \dots, T_m$ of sets of well-formed formulas $T_i \subseteq H$ is called a \emph{derivation} if: (1) $T_0 = A$; and (2) for each $i$ there exists some rule $r \in R$ of arity $n$ and some $\vec{t} = (t_1, \dots, t_n)$ with all $t_j \in T_i$ such that $T_{i+1} = T_i \cup \{r(\vec{t})\}$.
\end{definition}

Monotonic derivations are \emph{inflationary}~\citep{vianu_datalog_2021}: each step only \emph{adds} formulas, so the sets $T_i$ form a growing chain converging toward a least fixpoint.

\begin{definition}
\label{def:theorems-mono}
The \emph{set of theorems} $T$ of a formal system $(H, A, R)$ is the set of all formulas that occur in some derivation:
$$T = \bigcup_{\substack{T_0, \dots, T_m \\ \text{is a derivation}}} \bigcup_{i=0}^{m} T_i.$$
\end{definition}
It is well known that for finite systems, this set of theorems coincides with the least fixpoint of the immediate consequence operator~\citep{vanEmdenKowalski1976,abiteboul1995foundations}, providing the classic semantics for Datalog.


\subsection{Nonmonotonic Formal Systems}

In many practical situations, derived formulas must be \emph{retracted}. For example, when computing shortest paths, we may derive distance $k$ between two vertices, then find a shorter path of length $k' < k$ and want to replace the old value.

Various approaches model such retractions, e.g., via deletion rules~\citep{vianu_datalog_2021}. Here, we consider a \emph{noninflationary} approach where rules rewrite the entire state at each step (similar to noninflationary \emph{while} queries~\citep{abiteboul1995foundations}).

\begin{definition}
\label{def:nmfs}
A \emph{nonmonotonic formal system} is a triple $(H, A, W)$, where $H$ is a set of well-formed formulas, $A \subseteq H$ is a set of axioms, and $W$ is a set of rewrite rules. Each rewrite rule $w \in W$ maps sets of formulas to sets of formulas: $w: 2^H \rightarrow 2^H$.
\end{definition}

\begin{definition}
\label{def:derivation-nm}
For a nonmonotonic formal system $(H, A, W)$, a sequence $T_0, \dots, T_m$ with $T_i \subseteq H$ is called a \emph{derivation} if $T_0 = A$ and $T_{i+1} = w(T_i)$ for some $w \in W$.
\end{definition}
It is not immediately clear what the theorems of such a system should be, since derived formulas may be introduced and subsequently retracted as computation proceeds. We define truth based on the \emph{eventual behavior} of derivations.


\subsection{Thesis and Truth}

Rather than limiting attention to programs with fixpoints, we consider a semantics based on a \emph{thesis}: a claim about the eventual behavior of derivations.

\begin{definition}
\label{def:thesis}
Given a nonmonotonic formal system $(H, A, W)$, a \emph{thesis} is a set $M \subseteq 2^H$ of sets of formulas such that a set of formulas $T \subseteq H$ is \emph{consistent} with $M$ if $T \in M$, and otherwise is \emph{inconsistent} with $M$.
\end{definition}
Intuitively, a thesis $M$ represents a claim such as ``a deep enough derivation will arrive at a state in $M$''. To make this precise, we introduce the following graph-theoretic notions.

\begin{definition}
\label{def:state}
Given a nonmonotonic formal system $(H, A, W)$, a set $T \subseteq H$ is a \emph{state of derivation} if there exists a derivation of $T$, i.e., $T_0, \dots, T_{m-1}, T$.
\end{definition}
Derivation states form a derivation graph.
\begin{definition}\label{def:derivationgraph}
The \emph{graph of derivation} $G(H, A, W)$ is a directed graph whose vertices are states of derivation and whose edges are pairs $(T, T')$ such that $T' = w(T)$ for some $w \in W$.
\end{definition}
A state $T'$ is \emph{reachable} from $T$ if $T = T'$ or there is a finite sequence of edges in $G$ from $T$ to $T'$; that is, reachability is the reflexive--transitive closure of the edge relation.
Truth is defined over this graph.
\begin{definition}\label{def:truth}
A thesis $M$ \emph{holds} (equivalently: $M$ is \emph{true}, $M$ is a \emph{truth}) in a nonmonotonic formal system $(H, A, W)$ with derivation graph $G$ if for any derivation state $T_i$ there exists a derivation state $T_j \in M$ such that:
(1) $T_j$ is reachable from $T_i$; and (2) for every derivation state $T_k$ reachable from $T_j$, $T_k \in M$.
\end{definition}
Thus $M$ holds if no matter where you are in the derivation, you can steer toward a state in $M$, and once there, any further derivation stays within $M$. Truth permits \emph{stabilization}: derivation can always be
continued to achieve and maintain consistency with
the true thesis.

\begin{definition}
\label{def:theorem}
A formula $t \in H$ is a \emph{theorem} in $(H, A, W)$ if the thesis $\{T \subseteq H \mid t \in T\}$ holds. A formula $t$ is an \emph{anti-theorem} if the thesis $\{T \subseteq H \mid t \notin T\}$ holds. A formula that is either a theorem or an anti-theorem is called \emph{definite}.
\end{definition}
Note that a formula may be neither a theorem nor an anti-theorem, i.e., its status may be genuinely undetermined. It also follows that no formula can be both a theorem and an anti-theorem: if both held, then by Lemma~\ref{lemma:truth_conjunction} the empty thesis $\{T \mid t\in T\} \cap \{T \mid t\notin T\} = \emptyset$ would be a truth, which is impossible since no derivation state lies in $\emptyset$.


\subsection{Logica as a Nonmonotonic Formal System}
\label{sec:logica-is-nmfs}

We now define how a Logica program gives rise to a nonmonotonic formal system.

\begin{definition}
\label{def:logica-system}
Given a Logica program $\mathcal{P}$ over a database (EDB) $D$, the corresponding \emph{nonmonotonic formal system} $(H, A, W)$ is defined as follows:
\begin{itemize}
\item $H$, the Herbrand base, is the set of all ground atoms of the form $P(v_1, \ldots, v_n)$ for each predicate $P$ in $\mathcal{P}$ and all possible values. For functional predicates, we assume the last argument $v_n$ is the value argument.
\item $A = D$, the extensional database.
\item $W = \{w_P \mid P \text{ is an intensional predicate in } \mathcal{P}\}$, where each $w_P$ is a rewriting rule that recomputes all values of predicate $P$ based on the current state. Formally, $w_P(T) = (T \setminus \{P(\vec{v}) \in T\}) \cup \mathit{eval}_P(T)$, where $\mathit{eval}_P(T)$ evaluates all rules defining $P$ against $T$ (snapshot semantics): body literals are looked up in $T$ as-is, without recomputing other predicates. All non-$P$ atoms, including EDB facts, are unchanged.
\end{itemize}
\end{definition}
A key observation is that $W$ contains one rewriting rule per \emph{predicate}, not per syntactic rule. Consider again the program to compute shortest path distances \texttt{D}.
\begin{lstlisting}[style=logica,basicstyle=\small\ttfamily,columns=fullflexible]
    D(a,b) Min= E(a,b);
    D(a,b) Min= D(a,x) + D(x,b);
\end{lstlisting}
The two clauses in this case constitute a single rewriting rule $w_D$ that computes the minimum over all derivable values, both direct edges and composite paths. 
This design reflects a fundamental feature of aggregation: when $w_D$ is applied, it replaces the \emph{entire} predicate $D$ with newly computed values. If the current state contains \texttt{D("a","b") = 42} from an earlier derivation, but a shorter path of length $37$ is now derivable, the new state will contain \texttt{D("a","b") = 37}, where the old value is \emph{replaced}, not accumulated. This is precisely what makes the system nonmonotonic.

Each edge in the derivation graph (Def.~\ref{def:derivationgraph}) corresponds to applying one rewrite rule $w_P \in W$. At each step, we choose \emph{which predicate} to recompute, introducing nondeterminism. For single-predicate programs like shortest paths, only one rule applies. For programs with multiple interdependent predicates, different choices lead to different derivation graph trajectories.



\subsection{The Thesis Defense Game}

Truth admits a game-theoretic characterization that gives DO semantics its name.

\begin{definition}
The \emph{Thesis Defense game}\label{def:thesis_defense} for a nonmonotonic formal system $(H, A, W)$ and thesis $M$ is a two-player game between a \emph{Defendant} and an \emph{Opponent}. The game constructs a derivation in three turns that begins with the Opponent:
\begin{itemize}
\item \textbf{Turn 1 (Opponent):} Starting from $T_0 = A$, the Opponent extends the derivation to some state $T_{m_1}$.
\item \textbf{Turn 2 (Defendant):} The Defendant continues from $T_{m_1}$ to some state $T_{m_2}$.
\item \textbf{Turn 3 (Opponent):} The Opponent continues from $T_{m_2}$ to a final state $T_{m_3}$.
\end{itemize}
In a single turn a player may apply any finite number of rewrite rules, i.e., move to any state reachable from the current one; the available move sets may be infinite.
The Defendant wins if $T_{m_3} \in M$; the Opponent wins otherwise.
\end{definition}

\begin{theorem}\label{theorem:defence_game}
For a nonmonotonic formal system $(H, A, W)$ and thesis $M$, the following are equivalent:
\begin{enumerate}
\item Thesis $M$ holds in $(H, A, W)$.
\item The Defendant has a winning strategy.
\item The Opponent has no winning strategy.
\end{enumerate}
\end{theorem}

Equivalence of (2) and (3) holds by determinacy (the number of turns is fixed at three and there are no draws; the move sets may be infinite). Equivalence with (1) captures the essence of DO semantics: \emph{truth is what you can defend}.
The O-D-O turn order is essential: reversing it or reducing to two turns would break the closure of truths under intersection.

\begin{lemma}\label{lemma:truth_conjunction}
If $M_1$ and $M_2$ are truths, then $M_1 \cap M_2$ is also a truth.
\end{lemma}
\begin{proof}[Proof (Sketch)]
Since $M_1$ is true, the Defendant has a strategy to reach a state $T_j \in M_1$ that is stable within $M_1$. From $T_j$, since $M_2$ is also true, the Defendant can further reach a state $T_k \in M_2$ that is stable within $M_2$. Because $T_k$ is reachable from $T_j$ and $T_j$ is stable within $M_1$, we have $T_k \in M_1$ as well. Thus $T_k \in M_1 \cap M_2$, and any state reachable from $T_k$ remains in both $M_1$ and $M_2$.
\end{proof}


\subsection{Truth as Modality}

DO truth has a natural modal logic interpretation. The derivation graph is a Kripke frame, with derivation states as worlds and reachability as the accessibility relation.

\begin{definition}
\label{def:modal-statement}
Let $(H, A, W)$ be a nonmonotonic formal system. A \emph{modal statement} is defined inductively:
\begin{itemize}
\item For any formula $h \in H$ statements $h$ and $\neg h$ are modal statements.
\item If $t$ is a modal statement, then $\mabox\,t$ (read: ``necessarily $t$'') is a modal statement.
\item If $t$ is a modal statement, then $\madiamond\,t$ (read: ``possibly $t$'') is a modal statement.
\end{itemize}
\end{definition}

\begin{definition}
\label{def:modal-satisfaction}
Whether a modal statement $t$ holds at a state of derivation $T$, written $T \Vdash t$, is defined inductively, where any state $T$ is by definition considered reachable from itself:
\begin{itemize}
\item $T \Vdash h$ (for $h \in H$) if and only if $h \in T$.
\item $T \Vdash \neg h$ (for $h \in H$) if and only if $h \not \in T$.
\item $T \Vdash \mabox\,t$ if and only if $t$ holds at all states reachable from $T$.
\item $T \Vdash \madiamond\,t$ if and only if $t$ holds at some state reachable from $T$.
\end{itemize}
\end{definition}
We can now express DO-truth using a single modal-logic formula.
\begin{observation}
\label{obs:bdb}
In a nonmonotonic formal system $(H, A, W)$, a formula $t \in H$ is a theorem iff
$A \Vdash \bdb\, t$.
\end{observation}
Unpacking: from $A$, \emph{necessarily} (regardless of the Opponent's first move), it is \emph{possible} (the Defendant finds a way) to reach a state where $t$ \emph{necessarily} holds (persists under any further moves). The three alternations correspond to the three turns of the Thesis Defense game. The central $\madiamond$ is the diamond that guarantees forever: it ensures a path exists to a region where $t$ is permanently stable. In the framework of \citet{AlpernSchneider1985}, this stabilization condition decomposes into liveness (a stable region is reachable) and safety ($t$ persists once reached).
For monotonic systems, $A \Vdash \madiamond\, t$ suffices: once derived, a formula persists, so $\madiamond\, t$ implies $\bdb\, t$.

\begin{definition}
\label{def:monotonic-nmfs}
A nonmonotonic formal system $(H, A, W)$ is \emph{monotonic} if every rewrite rule is inflationary, i.e., $T \subseteq w(T)$ for all $w \in W$ and all states $T$ (equivalently, no atom is ever retracted). The monotonic formal systems considered earlier are the special case $w_r(T) = T \cup \{r(\vec{t})\}$.
\end{definition}

\begin{observation}
In a monotonic formal system $(H, A, W)$, for any $t \in H$, $A \Vdash \madiamond\, t$ iff $A \Vdash \bdb\, t$.
\end{observation}

In positive Datalog, derivation is monotonic: derived facts can only grow. The least fixpoint consists of atoms satisfying $A \Vdash \madiamond\, t$, which are exactly those satisfying $A \Vdash \bdb\, t$.

\begin{corollary}\label{cor:datalog}
For positive Datalog programs (without negation or nonmonotonic aggregation), DO semantics coincides with classical least-fixpoint semantics: a ground atom is a DO theorem iff it belongs to the least fixpoint.
\end{corollary}

The DO accessibility relation is reflexive and transitive, implying the modal logic S4. It is well known that in S4, arbitrary sequences of
$\mabox$ and $\madiamond$ collapse to just seven distinct modalities~\citep{Chellas_1980}, as shown on the left in Figure~\ref{fig:modalities}.
Figure~\ref{fig:modalities} (right) interprets these modalities: $t$ is \emph{derivable} ($\madiamond\, t$) if reachable from the axioms; \emph{provable} ($\madiamond\mabox\,t$) if some reachable state contains $t$ permanently; and \emph{always derivable} ($\mabox\madiamond\,t$) if from any state, a state with $t$ can be reached. DO-theoremhood ($\bdb$) is strictly stronger than each (but weaker than~$\mabox\,t$).

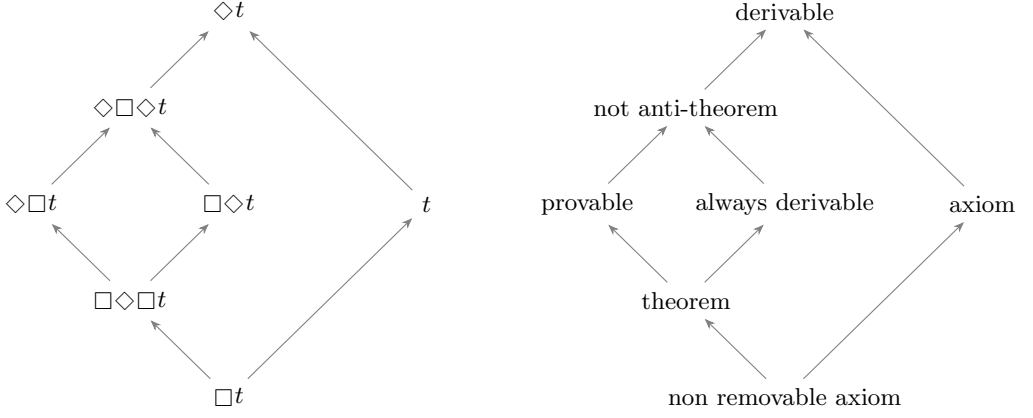
\begin{figure}[htbp]
\begin{center}
\resizebox{\textwidth}{!}{%
\begin{tikzpicture}[
  >={Stealth[length=4pt]},
  imp/.style={->, gray, shorten >=2pt, shorten <=2pt},
  m/.style={font=\small, inner sep=2pt},
  lab/.style={font=\small, align=center, inner sep=2pt},
  yscale=1.25,
  xscale=0.8
]
\node[m] (bt)   at (0,0)    {$\mabox t$};
\node[m] (bdbt) at (-1.6,1) {$\bdb t$};
\node[m] (dbt)  at (-3.2,2) {$\madiamond\mabox t$};
\node[m] (bdt)  at (0,2)    {$\mabox\madiamond t$};
\node[m] (t)    at (3.2,2)  {$t$};
\node[m] (dbdt) at (-1.6,3) {$\madiamond\mabox\madiamond t$};
\node[m] (dt)   at (0,4)    {$\madiamond t$};
\draw[imp] (bt) -- (bdbt);  \draw[imp] (bt) -- (t);
\draw[imp] (bdbt) -- (dbt); \draw[imp] (bdbt) -- (bdt);
\draw[imp] (dbt) -- (dbdt); \draw[imp] (bdt) -- (dbdt);
\draw[imp] (dbdt) -- (dt);  \draw[imp] (t) -- (dt);
\begin{scope}[xshift=9cm]
\node[lab] (Rbt)   at (0,0)    {non removable axiom};
\node[lab] (Rbdbt) at (-1.6,1) {theorem};
\node[lab] (Rdbt)  at (-3.2,2) {provable};
\node[lab] (Rbdt)  at (0,2)    {always derivable};
\node[lab] (Rt)    at (3.2,2)  {axiom};
\node[lab] (Rdbdt) at (-1.6,3) {not anti-theorem};
\node[lab] (Rdt)   at (0,4)    {derivable};
\draw[imp] (Rbt) -- (Rbdbt);  \draw[imp] (Rbt) -- (Rt);
\draw[imp] (Rbdbt) -- (Rdbt); \draw[imp] (Rbdbt) -- (Rbdt);
\draw[imp] (Rdbt) -- (Rdbdt); \draw[imp] (Rbdt) -- (Rdbdt);
\draw[imp] (Rdbdt) -- (Rdt);  \draw[imp] (Rt) -- (Rdt);
\end{scope}
\end{tikzpicture}%
}
\end{center}
\caption{\small The seven distinct S4 modalities (left, strongest $\mabox\,t$ at bottom) with their DO-semantics interpretations (right). Here $\mabox$ is ``necessarily'' and $\madiamond$ is ``possibly'' in the standard modal reading; the right side shows what each modality means for derivations. Arrows point upward in the direction of implication: a lower (stronger) modality entails a higher (weaker) one, e.g., $\bdb\,t$ (theorem) implies $\mabox\madiamond\,t$ (always derivable). Note: ``axiom'' means $t$ holds at the current state; ``non removable axiom'' ($\mabox\,t$) means $t$ holds at every reachable state.}
\label{fig:modalities}
\end{figure}


\section{Interpretations in DO Semantics}
\label{sec:interpretations}

Having defined which theses hold, we turn to interpretation. In classic logic programming, an interpretation is a model consistent with all theorems. For nonmonotonic systems, we distinguish \emph{static} and \emph{dramatic} interpretations.

\subsection{Static Interpretations}

\begin{definition}
\label{def:static-interp}
Let $(H, A, W)$ be a nonmonotonic formal system and let $\mathcal{T}$ be the set of all truths, that is all theses that hold in it. If a state of derivation $T$ belongs to the intersection of all truths
$T \in \bigcap_{M \in \mathcal{T}} M$, then $T$
is called a \emph{static interpretation} of the system.
\end{definition}

A static interpretation is a state of derivation consistent with \emph{every} truth: a stable resting point respecting all defensible claims. For finite systems, such interpretations always exist.

\begin{theorem} \label{theorem:static-sync}
Let $(H, A, W)$ be a nonmonotonic formal system with a finite derivation graph. 
Then static interpretations are precisely the nodes of the derivation graph
that belong to the terminal (sink) 
strongly connected components.
\end{theorem}

The finiteness assumption is essential: it ensures that every path in the derivation graph eventually reaches a terminal SCC. The proof follows from the observation that static interpretations are precisely the states $T$ for which the Opponent has a strategy to end the game on $T$.

These terminal SCCs are the analogue of \emph{attractors} in Boolean-network theory; the related notion of a \emph{trap space}, a region closed under the dynamics, captures the same persistence-under-derivation that DO expresses through the final $\mabox$ of $\bdb$. \citet{Trinh2025BooleanDatalog} establish this Boolean-network view of Datalog with negation (subset-minimal stable trap spaces correspond to regular models), while \citet{Trinh2025TrapSpace} develop a trap-space \emph{semantics} for normal logic programs.

\begin{definition}
\label{def:fixpoint}
A state of derivation $T \subseteq H$ is a \emph{fixpoint} if $w(T) = T$ for all rules $w \in W$.
\end{definition}

\begin{observation}
\label{obs:fixpoint-is-interpretation}
Any fixpoint is a static interpretation.
\end{observation}

\noindent While fixpoints are effectively sink states, static
interpretations may not be fixpoints (but they are consistent with all
truths).  See Figure~\ref{fig:derivation} for examples of static
interpretations.

\begin{figure}[t]
\centering
\begin{subfigure}[b]{0.36\textwidth}
  \centering
  \includegraphics[width=\textwidth]{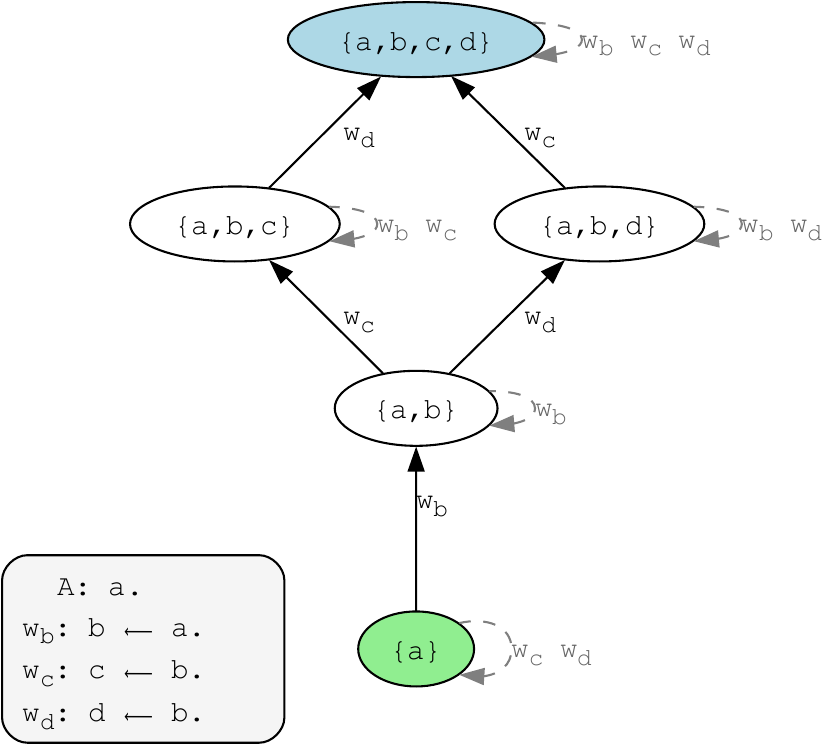}
  \label{fig:derivation-a}
\end{subfigure}
\hfill
\begin{subfigure}[b]{0.52\textwidth}
  \centering
  \includegraphics[width=\textwidth]{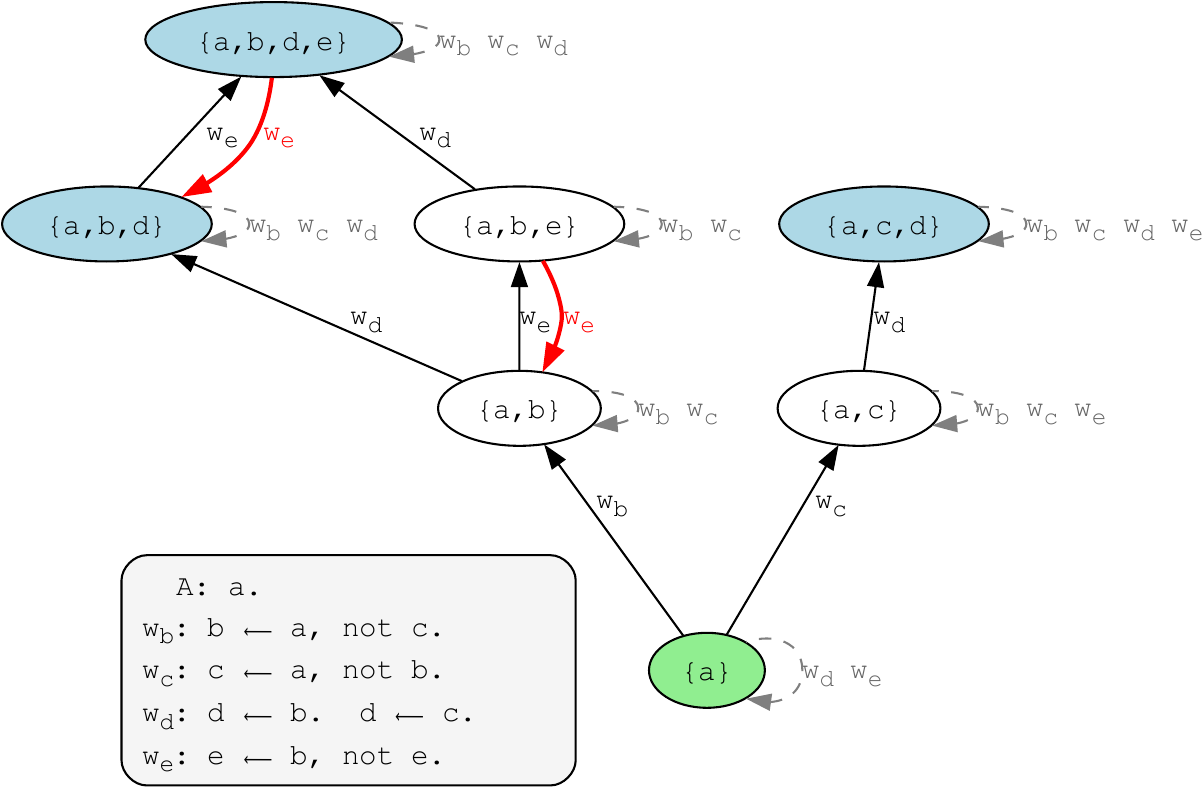}
  \label{fig:derivation-b}
\end{subfigure}
\caption{\small Two logic programs and their derivation graphs. 
Monotonic program (left): derivation flows upward to the unique fixpoint $\{a,b,c,d\}$.  Nonmonotonic program (right): $\mathtt{w_e}$ causes oscillation (red edges) between $\{a,b,d\}$ and
$\{a,b,d,e\}$; $\{a,c,d\}$ is a separate fixpoint.
Green nodes are axioms; blue nodes are static interpretations (sink SCC nodes). Edge labels identify rewrite rules; gray self-loops indicate rules that leave the state unchanged.}
\label{fig:derivation}
\end{figure}


\subsection{Dramatic Interpretations}

Static interpretations suffice when the system settles into a stable state. For programs like PageRank, where iteration converges \emph{toward} a value without reaching it, we need a dynamic counterpart.

\begin{definition}
\label{def:dramatic-interp}
A \emph{dramatic interpretation}\footnote{``Dramatic'' comes from the Ancient Greek word $\delta\rho\acute{\alpha}\omega$ referring to action, where a dramatic interpretation unfolds through action via the process of derivation itself.} of $(H, A, W)$ is an infinite sequence $T_1, T_2,\, \dots$ of derivation states such that for any thesis $M$ that holds, there is an $n$ where $T_k \in M$ for all $k > n$.
\end{definition}

A dramatic interpretation eventually enters and stays within every true thesis. Note that consecutive states $T_k, T_{k+1}$ need not be connected by a rewrite rule: a dramatic interpretation is any sequence of derivation states with this eventual-consistency property, not necessarily a path in the graph.

\begin{observation}\label{observe:static_in_dramatic}
If a dramatic interpretation $T_1, T_2,\,\dots$ eventually stabilizes (i.e., $T_k = T_n$ for all $k > n$), then $T_n$ is a static interpretation. Conversely, if $T$ is a static interpretation, then the constant sequence $T, T, T,\,\dots$ is a dramatic interpretation.
\end{observation}

Thus static interpretations are ``finished'' dramatic interpretations.

\subsection{$\omega$-Limit Interpretations}
\label{sec:omega}

For real-valued computation (PageRank, differential equations), we want derivations to converge topologically. This requires equipping the state space with a topology. The nondeterministic scheduling in DO semantics is related to asynchronous iteration~\citep{BertsekasTsitsiklis1989}, where convergence must hold regardless of update order.


\begin{definition}
\label{def:omega-limit}
Let $(H, A, W)$ be a nonmonotonic formal system and let $\mathcal{S} \subseteq 2^H$ be a set containing all states of derivation, equipped with a topology $\mathbb{T}$. A set $T \in \mathcal{S}$ is an \emph{$\omega$-limit interpretation} if every open set containing $T$ is a true thesis.
\end{definition}
For Logica programs with real-valued computation, we define a concrete topology.
\begin{definition}
\label{def:hausdorff}
Given a Herbrand base $H$ with arguments typed as strings, integers, or reals, we define a subset $\mathcal{S} \subseteq 2^H$ and a Hausdorff topology as follows.
Let $r$ be an injection from predicate symbols, strings, and integers into $\mathbb{N}$, extended to reals by the identity. Since argument positions do not mix types, $r$ is injective across all values in any given position. For a ground atom $P(v_1, \ldots, v_n)$, let $\rho(P(v_1, \ldots, v_n)) = (r(P), r(v_1), \ldots, r(v_n), 0, \ldots, 0) \in \mathbb{R}^{k}$, padded to dimension $k$ (the maximum predicate arity plus one).
The space $\mathcal{S}$ consists of all \emph{compact} subsets of $\mathbb{R}^k$, equipped with the Hausdorff distance:
$$d_H(T_1, T_2) = \max\left(\sup_{x \in T_1} \inf_{y \in T_2} d(\rho(x),\rho(y)), \sup_{y \in T_2} \inf_{x \in T_1} d(\rho(x),\rho(y))\right)$$
where $d$ is the Euclidean distance. The \emph{Hausdorff topology} on $\mathcal{S}$ is the topology induced by this metric.
\end{definition}

Note that per Definition~\ref{def:hausdorff} the
Herbrand base has arguments from $\mathbb{R}$ and is
thus of the cardinality of the continuum, while numbers used
in the definition of a logic program have only finitely many decimal digits. Thus, for instance, transcendental numbers (ones that are not roots of any
polynomial with rational coefficients)
never occur in states of derivation, because
they cannot be computed in a finite number of steps.
However,
transcendental numbers can occur in the $\omega$-limit, as we see
in the program converging to $\pi$ in Section~\ref{section:logica}.

All $\omega$-limit interpretations in this paper use the Hausdorff topology. States of derivation are finite, hence compact, hence elements of $\mathcal{S}$. Since predicate symbols map to distinct natural numbers, facts from different predicates remain at distance ${\ge}\,1$, so each predicate's convergence can be analyzed independently.

\begin{observation}
If an $\omega$-limit interpretation exists, then it is unique and every dramatic interpretation converges to it. (This follows from the Hausdorff separation property.)
\end{observation}

\begin{theorem}\label{thm:pagerank}
The PageRank program from Section~\ref{section:logica} has an $\omega$-limit interpretation, which assigns to each page its true PageRank value, i.e., the eigenvector of the link matrix.
\end{theorem}

\begin{proof}
Applying $w_{PageRank}$ corresponds to multiplying the current PageRank vector by the transition matrix $M = (1-d)\frac{1}{n}\mathbf{1} + d \cdot L$, where $d$ is the damping factor and $L$ the column-normalized link matrix (assuming no dangling nodes, or with the standard dangling-mass correction, so that $L$ is column-stochastic). Since $M$ is stochastic with all positive entries, the Perron-Frobenius theorem~\citep{meyer2000matrix} gives a unique eigenvector $v^*$ with eigenvalue 1, and the power method converges to $v^*$ from any starting distribution.

For any open neighborhood $U$ of $v^*$, the Defendant can apply $w_{PageRank}$ sufficiently many times to enter $U$, and all subsequent iterations remain within $U$ (they only get closer to $v^*$). Thus every neighborhood of $v^*$ is a true thesis, making $v^*$ an $\omega$-limit interpretation.
\end{proof}

By a similar argument, the $\pi$-approximation program has an $\omega$-limit interpretation with $\mathtt{Pi()} = \pi$ (the Leibniz series converges, so the Defendant can increase $L$ until the partial sum is within any $\varepsilon$ of~$\pi$).

\begin{theorem}\label{thm:sincos}
The differential equation program for sine and cosine, restricted to $t \leq 2\pi$, has an $\omega$-limit interpretation where $\mathtt{S}$ and $\mathtt{C}$ equal the sine and cosine functions on $[0, 2\pi]$.
\end{theorem}

\begin{proof}
The program implements Euler's method for $S' = C$, $C' = -S$ with $S(0) = 0$, $C(0) = 1$, whose unique solution is $S = \sin$, $C = \cos$. Each rewriting rule is idempotent when its dependencies are unchanged, so no progress is lost by repeated application.

Let $U$ be any open neighborhood of $(\sin, \cos)$ on $[0, 2\pi]$. Euler's method converges uniformly on compact sets~\citep{atkinson2009numerical}: for any $\varepsilon > 0$, there exists $L_0$ such that step size $1/L_0$ stays within $\varepsilon$ of the true solution.

The Defendant first applies $w_L$ until $L = L_0$ (choosing $\varepsilon$ so the $\varepsilon$-neighborhood lies inside~$U$), then alternates $w_S$ and $w_C$ to compute the full solution on $\{0, 1/L_0, 2/L_0, \ldots\}$, reaching a state within~$U$.

On the Opponent's final turn: applying $w_S$ or $w_C$ is idempotent (dependencies unchanged); applying $w_L$ increments $L$ to $L' > L_0$, which only \emph{improves} accuracy, so the state remains within~$U$. Thus every open neighborhood of $(\sin,\cos)$ is a true thesis.
\end{proof}

\subsection{Irreducible Truth}
\label{app:irreducible-truth}

In general, static interpretations do not capture the full
structure of truth in a nonmonotonic formal system.
As discussed, a meaningful program (for example PageRank)
may have no static interpretation. But even when a static interpretation exists, a distinct dramatic
interpretation may exist alongside. The program below illustrates this:
\begin{lstlisting}[style=logica,basicstyle=\small\ttfamily,columns=fullflexible]
    Left() :- not Right();
    Right() :- not Left();
    X() = 1 :- not X(), Right();
    X() = -1 :- not X(), Left();
    X() = X() + 1 :- X() > 0;
    X() = X() :- X() < 0;
\end{lstlisting}
If $\mathtt{Left}$ holds, then $\mathtt{X}$ stabilizes at $-1$, yielding a static interpretation $\{\mathtt{Left}, \mathtt{X}(-1)\}$.
If $\mathtt{Right}$ holds, then $\mathtt{X}$ grows without bound ($1, 2, 3, \ldots$) and no fixpoint is reached.
Since derivation states in the $\mathtt{Right}$ branch cannot reach $\{\mathtt{Left}, \mathtt{X}(-1)\}$, the set of static interpretations is not reachable from every state and thus is not a truth.

\begin{figure}[t]
\centering
\includegraphics[width=0.33\textwidth]{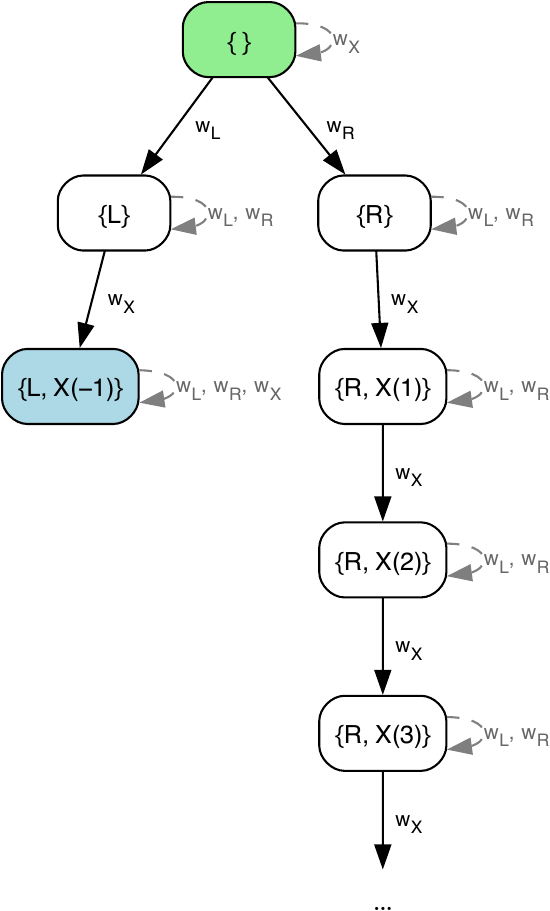}
\caption{\small Derivation graph for the Left/Right program (L and R
  abbreviate Left and Right).  Green: axiom; blue: sole static
  interpretation (fixpoint).  The left branch stabilizes; the right
  diverges ($X$ grows without bound).  Since the blue node is
  unreachable from the right branch, no irreducible truth exists.
  Dashed gray: self-loops.}
\label{fig:irreducible}
\end{figure}

For finite systems, however, static interpretations do
characterize all truths. We formalize this via the notion
of irreducible truth.

\begin{definition}
\label{def:irreducible-truth}
Let $(H, A, W)$ be a nonmonotonic formal system.
If the set of all static interpretations is itself a truth,
then it is called the
\emph{irreducible truth} of the system.
\end{definition}

\begin{observation}
\label{obs:irreducible-characterizes}
If the irreducible truth exists, then a thesis holds
if and only if it contains the irreducible truth,
i.e., every static interpretation belongs to the thesis.
\end{observation}

\begin{proof}
Any thesis that contains a truth is itself a truth:
if from every state one can reach and stay in $M$,
then one can also reach and stay in any superset of $M$.
Thus any thesis containing the irreducible truth holds.
Conversely, if a thesis holds, it must contain every
static interpretation: a static interpretation belongs
to every truth, hence to every thesis that holds.
Therefore it contains the irreducible truth.
\end{proof}

By definition, whenever the irreducible truth exists, any static interpretation is an element of the irreducible truth. The following observation shows that dramatic interpretations are also characterized by the irreducible truth.

\begin{observation}
    If the irreducible truth exists, then
    dramatic interpretations are precisely sequences of
    derivation states $T_1, \dots, T_k, \dots$
    for which there exists $n$ such that for any $k > n$
    we have $T_k$ being an element of the irreducible truth.
\end{observation}
\begin{proof}
    Consider a dramatic interpretation $T_1, \dots, T_k, \dots$.
    Since the irreducible truth is a truth, by definition of a
    dramatic interpretation there exists $n$ such that
    states $T_{n+1}, \dots$ belong to the irreducible truth.
    The converse also holds because a tail of elements of the
    irreducible truth is consistent with every truth.
\end{proof}

\begin{observation}
\label{obs:finite-irreducible}
If the derivation graph is finite, then the irreducible
truth exists and equals the set of nodes in the terminal
(sink) strongly connected components.
\end{observation}

\begin{proof}
In a finite derivation graph, every path eventually
reaches a terminal SCC. Thus from any state, there
exists a path to a node in a terminal SCC. Once in a
terminal SCC, all further states remain within terminal
SCCs, since there are no outgoing edges to non-terminal
nodes. Therefore the set of terminal SCC nodes
is a truth.
By Theorem~\ref{theorem:static-sync}, this set is
exactly the set of static interpretations,
so it is the irreducible truth.
\end{proof}


\section{Connections to Classical Nonmonotonic Semantics}
\label{sec:connections}

We now examine how DO semantics relates to Well-Founded Semantics (WFS) and Stable Model Semantics (SMS).
The relationship is one of compatibility rather than equivalence:
DO semantics partially aligns with WFS and SMS, but differs on programs that fall outside their intended scope.
This asymmetry is by design.
WFS and SMS are defined for programs that reach fixpoints,
while DO semantics also assigns meaning to programs
that converge without reaching one,
such as PageRank, numerical approximation of $\pi$,
and differential equations solved via Euler's method
(Section~\ref{sec:omega}).

\begin{figure}[t]
\centering
\begin{subfigure}[c]{0.22\textwidth}
  \centering
  \includegraphics[width=\textwidth]{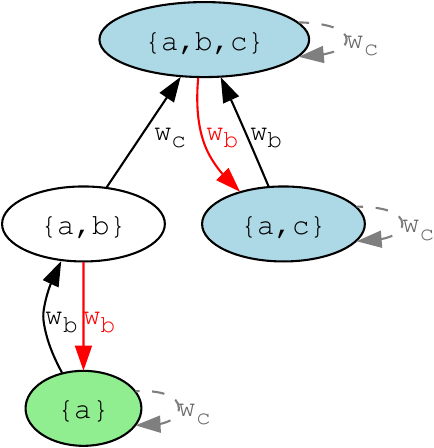}
  \label{fig:counter-wfs}
\end{subfigure}
\hfill
\begin{subfigure}[c]{0.75\textwidth}
  \centering
  \includegraphics[width=\textwidth]{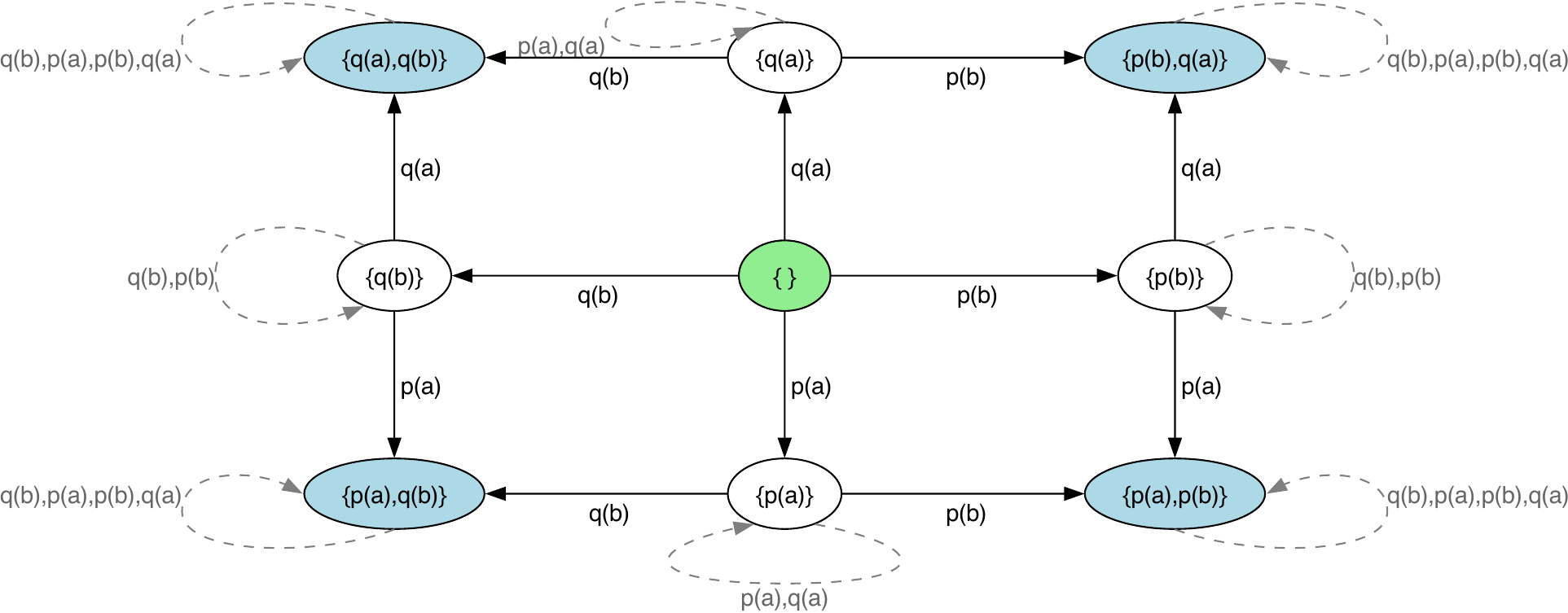}
  \label{fig:ground-nmfs}
\end{subfigure}
\caption{\small Left: derivation graph for the program with rules
  $\mathtt{w_b{:}\,\{b \leftarrow a, \neg\, b\}}$ and
  $\mathtt{w_c{:}\, \{ c \leftarrow b; ~ c \leftarrow c\}}$.  Right: Ground
  nonmonotonic system with rules
  $\mathtt{w_{p(a)}{:}\, \{p(a) \leftarrow \neg q(a)\}}$;
  $\mathtt{w_{p(b)}{:}\, \{p(b) \leftarrow \neg q(b)\}}$;
  $\mathtt{w_{q(a)}{:}\, \{q(a) \leftarrow \neg p(a)\}}$; and
  $\mathtt{w_{q(b)}{:}\, \{ q(b) \leftarrow \neg p(b)\}}$.  In both figures,
  green nodes represent axioms ($\mathtt{\{a\}}$ and $\emptyset$,
  respectively); blue nodes are static interpretations; dashed gray loops
  are rule applications that leave the state unchanged; red edges mark
  oscillation (caused by $\mathtt{w_b}$ on the left).
}
\label{fig:do-vs-ground}
\end{figure}

The program and its derivation graph shown on the left in Figure~\ref{fig:do-vs-ground} illustrate a key distinction of DO semantics.
The rule $b \leftarrow a, \neg\, b$ is self-defeating:
applying $\mathtt{w_b}$ alternately introduces and removes $b$.
However, when $b$ is present, $c \leftarrow b$ fires,
and once derived, $c$ persists via $c \leftarrow c$.
WFS declares only $a$ definite and leaves $b$ and $c$ undefined.
SMS finds no stable model at all, since $b$'s self-defeating
nature prevents any self-consistent fixpoint.
DO semantics, by contrast, identifies $\{a,c\}$ and $\{a,b,c\}$
as static interpretations and recognizes $c$ as a theorem:
from any state, derivation has a chance to reach the sink SCC
where $c$ is present in every state.
This willingness to let rules fire and retain what stabilizes
is precisely what allows DO semantics to give meaning
to convergent computations like PageRank.

Logica's semantics (Definition~\ref{def:logica-system}) uses
per-predicate rewrite rules, while established logic programming
semantics are defined on ground programs.
To connect the two, we define the ground nonmonotonic system
of a logic program.

\begin{definition}\label{def:ground-do}
Given a normal logic program~$P$, the {\em ground nonmonotonic system}
of~$P$ is the
nonmonotonic formal system $(H, A, W_{\mathrm{gr}})$ where
$H$ is the Herbrand base of~$P$,
$A$ is the set of facts of~$P$,
and $W_{\mathrm{gr}}$ contains one rewrite rule $w_{P(\vec{v})}$
per ground atom $P(\vec{v}) \in H$.
Rule $w_{P(\vec{v})}$ ensures that $P(\vec{v})$ belongs to the resulting state
iff some rule with head $P(\vec{v})$ has its body satisfied
in the current state, leaving all other atoms unchanged.
\end{definition}

The ground nonmonotonic system is not a separate semantics: it instantiates the same nonmonotonic formal system with finer-grained, per-ground-atom rewrite rules, serving as a bridge to the ground structure assumed by WFS and SMS.

On the right of Figure~\ref{fig:do-vs-ground} we show the
derivation graph of a ground nonmonotonic system for a program
with two predicates \verb|p| and \verb|q|.
The ground system has four static interpretations (blue nodes),
including mixed states such as $\{p(a), q(b)\}$ that are
unreachable under per-predicate rewrite rules.

\subsection{Well-Founded Semantics}

WFS~\citep{van1991well} assigns each ground atom to \emph{true},
\emph{false}, or \emph{undefined} via an alternating fixpoint
construction~\citep{vanGelder1993wfs}. The key observation is that
the result of this construction can be reproduced as a derivation
strategy within DO semantics, reaching a state $T_\mathit{WFS}$
containing exactly the WFS-true atoms.
From $T_\mathit{WFS}$, no further rule application can derive
WFS-false atoms. In modal terms, $T_\mathit{WFS}$ satisfies
$\madiamond\mabox\, t$ for every WFS-true atom $t$ and
$\madiamond\mabox\, \neg f$ for every WFS-false atom $f$.

\mypara{Example (Win-Move)}
Consider the win-move game~\citep{vanGelder1993wfs}:
\begin{lstlisting}[style=logica,basicstyle=\small\ttfamily,columns=fullflexible]
    Win(x) :- Move(x,y), not Lose(y);
    Lose(x) :- Position(x), not Win(x);
\end{lstlisting}
Under Logica's DO semantics, derivation alternates between recomputing $\mathit{Win}$ and $\mathit{Lose}$ (each recomputation is idempotent). Since the derivation graph is finite, every infinite derivation enters a terminal SCC. Atoms (e.g., $\mathit{Win}(x)$) present in every state of every terminal SCC correspond to WFS-true atoms; atoms in none correspond to WFS-false; atoms that oscillate correspond to WFS-undefined.
\begin{theorem}\label{thm:wfs}
Let $P$ be a logic program and $(H, A, W_{\mathrm{gr}})$
its ground nonmonotonic system.
Let $t$ be a ground atom.
\begin{itemize}
\item If $t$ is true under WFS, then $t$ is not an anti-theorem
      in $(H, A, W_{\mathrm{gr}})$.
\item If $t$ is false under WFS, then $t$ is not a theorem
      in $(H, A, W_{\mathrm{gr}})$.
\end{itemize}
\end{theorem}

\begin{proof}[Proof (Sketch)]
The alternating-fixpoint construction produces the WFS-true atoms as an increasing sequence $K_1 \subseteq K_2 \subseteq \cdots$: an atom enters $K_{i+1}$ only when it is derivable using, positively, atoms already established true, and, negatively, only atoms already established \emph{false} (atoms that can never become WFS-true). The Opponent's strategy is to fire the rewrite rule $w_{P(\vec{v})}$ for each WFS-true atom $P(\vec{v})$ in this order of first appearance (ties within a stage in any order). Negative conditions invoked at each stage refer only to already-falsified atoms, so the Opponent never relies on the absence of an atom that is in fact WFS-true. This reaches a state $T_\mathit{WFS}$ containing exactly the WFS-true atoms. From $T_\mathit{WFS}$ no rule application can introduce a WFS-false atom (its body would require a WFS-true atom to be absent or a WFS-false atom to be present, neither of which holds), and every WFS-true atom is self-supporting; hence every reachable state retains all WFS-true atoms and excludes all WFS-false atoms. Therefore no WFS-true atom is an anti-theorem and no WFS-false atom is a theorem.
\end{proof}

Theorem~\ref{thm:wfs} does \emph{not} claim WFS-true atoms are DO theorems. The relationship is compatibility: the two never contradict, but each may be definite where the other is~not.


\subsection{Stable Model Semantics}

SMS~\citep{gel1988stable} defines a stable model $M$ as the least model of the reduct $P^M$ and forms the basis of Answer Set Programming~\citep{Gelfond08}.

\begin{theorem}\label{thm:stable}
Let $P$ be a logic program and $(H, A, W_{\mathrm{gr}})$
its ground nonmonotonic system.
Every stable model of $P$ is a static interpretation
of $(H, A, W_{\mathrm{gr}})$.
\end{theorem}

\begin{proof}[Proof (Sketch)]
Let $M$ be a stable model of~$P$. We show that $M$ is a fixpoint of $(H, A, W_{\mathrm{gr}})$.
For any ground atom $P(\vec{v})$, the rewrite rule $w_{P(\vec{v})}$
evaluates all rules with head $P(\vec{v})$ against~$M$,
with $\mathit{not}\; Q(\vec{x})$ true iff $Q(\vec{x}) \notin M$.
This is exactly the evaluation under the reduct~$P^M$.
If $P(\vec{v}) \in M$, then since $M = \mathit{lfp}(T_{P^M})$,
some rule body is satisfied, so $w_{P(\vec{v})}$ retains $P(\vec{v})$.
If $P(\vec{v}) \notin M$, then no rule body is satisfied,
so $w_{P(\vec{v})}$ does not introduce $P(\vec{v})$.
In both cases $w_{P(\vec{v})}(M) = M$,
so $M$ is a fixpoint and therefore a static interpretation
by Observation~\ref{obs:fixpoint-is-interpretation}.
\end{proof}
The converse does not hold, i.e., DO admits interpretations that are not stable models. For programs having stable models, DO agrees; but DO
also assigns meaning to programs where stable models are absent or
where convergent iteration matters (as in PageRank).

\section{Conclusion}

We introduced Defendant-Opponent (DO) semantics, a stabilization-based framework for nonmonotonic logic programs with unrestricted recursion and aggregation. DO-truth admits three equivalent characterizations: graph-theoretic (stabilization in the derivation graph), game-theoretic (the Thesis Defense game), and modal (the $\bdb$ condition in S4). DO semantics coincides with least fixpoint semantics for positive Datalog, is compatible with WFS and SMS, and extends to programs that converge without reaching a fixpoint. Through $\omega$-limit interpretations, it bridges discrete symbolic reasoning with continuous mathematics, giving formal meaning to programs that compute real-valued limits.


\bibliographystyle{plainnat}
\bibliography{refs}


\ifappendix
\clearpage
\appendix
\begingroup\centering
  {\large\bfseries Supplementary Material}\\[3pt]
  \emph{for ``Diamonds Are Forever: Stabilization Semantics for
  Unrestricted Aggregation and Recursion in Logica''}\par
\endgroup
\phantomsection\label{app:start}
\medskip

\noindent All numbered cross-references (e.g., Definition~\ref{def:ground-do},
Theorem~\ref{thm:wfs}) refer to the main paper.
\medskip

\setcounter{definition}{0}
\renewcommand{\thedefinition}{S\arabic{definition}}

\section{Formal Syntax of Logica Programs}
\label{app:syntax}

The following definitions describe the formal syntax of the
fragment of Logica used in this paper. This is a strict subset
of the full Logica language, which includes additional features
(named arguments, meta-programming, imports) for practical
convenience. The fragment defined here is sufficient to express
all programs discussed in this paper.

\begin{definition}[Numbers]
A \emph{number} is an integer or a real number.
An \emph{integer} is a finite sequence of decimal digits,
optionally preceded by a minus sign.
A \emph{real number} is an integer followed by a decimal point
and a finite sequence of decimal digits.
\end{definition}

\begin{definition}[Strings]
A \emph{string} is a finite sequence of characters enclosed
in double quotes.
\end{definition}

\begin{definition}[Literals]
A \emph{literal} is one of the following:
\begin{enumerate}
\item a number,
\item a string,
\item the symbol \texttt{nil},
\item a predicate name (an identifier beginning with a capital letter),
\item a list $[e_1, \ldots, e_n]$ where each $e_i$ is an expression.
\end{enumerate}
\end{definition}

\begin{definition}[Expressions]
An \emph{expression} is one of the following:
\begin{enumerate}
\item a literal,
\item a variable (an identifier beginning with a lowercase letter),
\item an arithmetic expression $e_1 \mathbin{\mathit{op}} e_2$
      where $e_1$ and $e_2$ are expressions and
      $\mathit{op} \in \{+, -, *, /, \texttt{\^{}}\}$,
\item a predicate call $P(e_1, \ldots, e_n)$ where $P$ is a
      predicate name and each $e_i$ is an expression,
\item an aggregating expression (Definition~\ref{def:agg-expr}).
\end{enumerate}
\end{definition}

\noindent
A small set of \emph{built-in predicates}---for example \texttt{Range},
\texttt{Sqrt}, \texttt{N}, \texttt{Degree}, and \texttt{ResetProb}---is assumed
primitive and may appear wherever a predicate call is allowed; for instance,
\texttt{Range(n)} evaluates to the list $[0, \ldots, n-1]$.

\begin{definition}[Aggregating Expressions]
\label{def:agg-expr}
An \emph{aggregating expression} has the form
$\mathit{Agg}\{e \mathbin{:\text{-}} \varphi\}$
where:
\begin{enumerate}
\item $\mathit{Agg}$ is an aggregation operator
      (e.g., \texttt{Sum}, \texttt{Count}, \texttt{Min},
      \texttt{Max}, \texttt{List}),
\item $e$ is an expression,
\item $\varphi$ is a proposition.
\end{enumerate}
\end{definition}

\begin{definition}[Propositions]
A \emph{proposition} is one of the following:
\begin{enumerate}
\item a predicate call $P(e_1, \ldots, e_n)$,
\item a comparison $e_1 \mathbin{\mathit{op}} e_2$ where $e_1$
      and $e_2$ are expressions and
      $\mathit{op} \in \{=, \neq, <, >, \leq, \geq\}$,
      where $=$ serves as both equality test and
      unification (as in Prolog),
\item a membership test $e \in E$ where $e$ is an expression
      and $E$ is an expression evaluating to a list,
\item a negation $\mathit{not}\;\varphi$ where $\varphi$ is a
      proposition,\footnote{In the actual Logica syntax,
      negation is written as \texttt{\~{}}.}
\item a conjunction $\varphi_1 \land \varphi_2$ where $\varphi_1$
      and $\varphi_2$ are propositions.
\end{enumerate}
\end{definition}

\begin{definition}[Rules]
A \emph{rule} has one of the following forms:
\begin{enumerate}
\item \emph{Relational:}
      $P(e_1, \ldots, e_n) \mathbin{:\text{-}} \varphi$
\item \emph{Functional:}
      $P(e_1, \ldots, e_n) = e \mathbin{:\text{-}} \varphi$
\item \emph{Aggregated functional:}
      $P(e_1, \ldots, e_n) \mathbin{\odot}= e \mathbin{:\text{-}} \varphi$
\end{enumerate}
where:
\begin{itemize}
\item $P$ is a predicate name,
\item $e_1, \ldots, e_n$ and $e$ are expressions,
\item $\varphi$ is a proposition
      (which may be omitted, denoting $\mathit{true}$),
\item $\odot$ is an aggregation operator.
\end{itemize}
\end{definition}

\begin{definition}[Logica Program]
A \emph{Logica program} is a finite set of rules.
All rules defining a predicate $P$ must share the same
\emph{signature}: the same number of arguments, either all
with or all without a functional value, and if aggregated,
all with the same aggregation operator~$\odot$.
Disjunction is expressed implicitly: multiple rules with
the same predicate in the head represent alternative
derivations.
\end{definition}

\pagebreak

\section{Glossary of Key Terms}
\label{app:glossary}

\begin{description}

\item[Formal system] $(H, A, R)$:
  a set of formulas $H$, axioms $A$, and inference rules $R$
  (Definition~\ref{def:formal-system}).

\item[Nonmonotonic formal system] $(H, A, W)$:
  a set of formulas $H$, axioms $A$, and rewrite rules $W$
  that may retract derived formulas (Definition~\ref{def:nmfs}).

\item[Derivation:]
  a sequence of states starting from axioms $A$,
  where each step applies a rewrite rule
  (Definition~\ref{def:derivation-nm}); equivalently, a path in the
  derivation graph. (Dramatic interpretations, by contrast, need not be
  edge-connected; see their entry below.)

\item[State of derivation:]
  any set of formulas reachable by some derivation
  (Definition~\ref{def:state}).

\item[Derivation graph:]
  directed graph whose nodes are states of derivation
  and edges correspond to rewrite rule applications
  (Definition~\ref{def:derivationgraph}).

\item[Thesis:]
  a set $M$ of states representing a claim about
  derivation behavior (Definition~\ref{def:thesis}).

\item[Truth:]
  a thesis $M$ is a truth (i.e., holds) if from any state
  one can reach a state in $M$ from which all further
  states remain in $M$ (Definition~\ref{def:truth}).

\item[Theorem:]
  a formula $t$ such that the thesis $\{T \mid t \in T\}$
  holds; equivalently, $A \Vdash \Box\Diamond\Box\, t$
  (Definition~\ref{def:theorem}, Observation~\ref{obs:bdb}).

\item[Anti-theorem:]
  a formula $t$ such that the thesis $\{T \mid t \notin T\}$
  holds; equivalently, $A \Vdash \bdb\,\neg t$. Its negation,
  ``not an anti-theorem,'' corresponds to the dual modality
  $\madiamond\mabox\madiamond\, t$ (Definition~\ref{def:theorem}).

\item[Definite:]
  a formula that is either a theorem or an anti-theorem
  (Definition~\ref{def:theorem}).

\item[Static interpretation:]
  a state of derivation consistent with every truth;
  in finite systems, precisely the nodes in terminal
  SCCs of the derivation graph
  (Definition~\ref{def:static-interp}, Theorem~\ref{theorem:static-sync}).

\item[Fixpoint:]
  a state unchanged by all rewrite rules;
  every fixpoint is a static interpretation
  (Definition~\ref{def:fixpoint}, Observation~\ref{obs:fixpoint-is-interpretation}).

\item[Dramatic interpretation:]
  an infinite sequence of states such that for each true
  thesis $M$, there exists a step after which all states
  in the sequence belong to $M$ (Definition~\ref{def:dramatic-interp}).

\item[$\omega$-limit interpretation:]
  a topological limit: a state whose every open
  neighborhood is a true thesis (Definition~\ref{def:omega-limit}).

\item[Thesis Defense game:]
  a three-turn game (Opponent--Defendant--Opponent);
  a thesis holds iff the Defendant has a winning strategy
  (Definition~\ref{def:thesis_defense}, Theorem~\ref{theorem:defence_game}).

\item[Ground nonmonotonic system] $(H, A, W_{\mathrm{gr}})$:
  instantiation with one rewrite rule per ground atom,
  used for connection to WFS and SMS (Definition~\ref{def:ground-do}).

\end{description}
\fi

\end{document}